\documentclass[a4paper,UKenglish,cleveref, autoref, thm-restate]{lipics-v2021}

\pdfoutput=1 
\hideLIPIcs  

\title{A Primal-Dual Approach to Randomized Online Bidding with Tail Constraints} 

\author{Royce Kraakman}{Utrecht University}{}{}{}
\author{Bob Krekelberg}{Utrecht University}{}{}{}
\author{Alison Hsiang-Hsuan Liu}{Utrecht University}{}{}{}
\author{Fu-Hong Liu}{Independent Researcher}{}{}{}

\authorrunning{R. Kraakman, B. Krekelberg, A. H. Liu and F. Liu}

\Copyright{Royce Kraakman, Bob Krekelberg, Alison Hsiang-Hsuan Liu and Fu-Hong Liu}

\ccsdesc[500]{Theory of computation~Online algorithms}

\keywords{Online bidding, Tail constraints, Primal-dual approach, Randomized algorithms, Expected competitive ratio} 

\nolinenumbers 

\EventEditors{John Q. Open and Joan R. Access}
\EventNoEds{2}
\EventLongTitle{42nd Conference on Very Important Topics (CVIT 2016)}
\EventShortTitle{CVIT 2016}
\EventAcronym{CVIT}
\EventYear{2016}
\EventDate{December 24--27, 2016}
\EventLocation{Little Whinging, United Kingdom}
\EventLogo{}
\SeriesVolume{42}
\ArticleNo{23}

\usepackage{amsmath}
\usepackage{amssymb}
\usepackage{mathtools}
\usepackage[most]{tcolorbox}
\usepackage{amsthm}

\usepackage{graphicx}
\usepackage{subcaption}

\usepackage{hyperref}

\newcommand{\dA}{A} 
\newcommand{\A}{\mathcal{A}} 
\newcommand{\OPT}{\text{OPT}}
\newcommand{\cost}{\operatorname{cost}}
\newcommand{\E}{\mathbb{E}}
\newcommand{\runtitle}[1]{\textbf{#1}}
\newcommand{\hide}[1]{}
\newcommand{\gmlb}{2.8311}

\newcommand{\arxivonly}[1]{#1}

\makeatletter
\@fleqnfalse
\@mathmargin\@centering
\makeatother

\begin{document}

\maketitle

\begin{abstract}
    Controlling tail risk in randomized algorithms has received increasing attention.
    Randomized online algorithms can achieve a lower expected competitive ratio than the competitive ratio of the best deterministic online algorithm. However, individual executions may still incur poor outcomes. Tail constraints address this risk by limiting the probability that the incurred ratio exceeds a prescribed threshold. We investigate how such constraints affect the structure and performance of randomized online bidding algorithms.

    By revisiting the classical primal-dual lower-bound framework for online bidding, we find that restricting the growth between consecutive bids can yield a stronger lower bound on the expected competitive ratio. We then show that tail constraints with zero exceedance probability impose an asymptotic upper bound on the ratios between consecutive bids. Incorporating this structural property into the primal-dual framework allows us to derive a stronger lower bound. Moreover, the dual analysis suggests a randomized geometric bidding algorithm whose performance matches this bound, providing an exact characterization of the optimal expected competitive ratio under zero-violation tail constraints.

    The insight from our primal-dual analysis also guides the design of algorithms for general tail constraints allowing positive exceedance probabilities. Motivated by the geometric growth suggested by the dual solution, we study randomized periodic bidding algorithms. We show that, for any fixed geometric base, the uniform offset distribution is optimal in terms of both the expected competitive ratio and tail-constraint feasibility. We then characterize the exceedance probability of these algorithms and use this characterization to choose a geometric growth factor that satisfies the prescribed tail constraints, establishing parameter-dependent upper bounds on the optimal expected competitive ratio.
    
    \hide{Our primal-dual analysis also provides insight into the design of the matching randomized algorithm. For general tail constraints that allow positive violation probabilities, we further establish parameter-dependent upper bounds on the optimal expected competitive ratio.}

    \hide{A recently introduced approach is to impose tail constraints that limit the probability of poor outcomes.
    We study the randomized online bidding problem under such tail constraints.
    When the tail constraints require zero probability of exceeding the prescribed thresholds, we determine the optimal expected competitive ratio.
    For general tail constraints that allow positive exceedance probabilities, we derive a parameter-dependent upper bound on the optimal expected competitive ratio.}
    %
\end{abstract}

\section{Introduction}
In the study of online problems, \emph{competitive analysis} is the standard measure for evaluating deterministic online algorithms. It considers the worst-case ratio between the cost of the algorithm and the optimal offline cost over all possible instances.

Randomized online algorithms allow us to go beyond deterministic worst-case guarantees. Conceptually, a randomized online algorithm runs different deterministic online algorithms in different parallel universes according to some probability distribution. Its \emph{expected competitive ratio} measures, for the worst possible instance, the average performance across all these parallel universes.

\medskip

\runtitle{Controlling tail risk.}
Although a good randomized online algorithm performs well on average, consider the particular universe in which we happen to live. The deterministic algorithm running in this universe may still perform badly on a given instance. Average happiness is little comfort if the universe we inhabit is unhappy.

A \emph{tail constraint} limits the probability of such an unhappy outcome. More precisely, a tail constraint $(\gamma, \delta)$ bounds the probability that the competitive ratio exceeds a prescribed threshold $\gamma$ by $\delta$, while we still seek to minimize the randomized algorithm's expected competitive ratio.

\medskip

\runtitle{Necessity of ``bad'' deterministic algorithms.}
Consider all possible deterministic algorithms $\mathcal{D} = \{\text{DET}_i\}_{i\geq 1}$.
A randomized algorithm $\A$ can be considered as a distribution on $\mathcal{D}$.
The expected competitive ratio of $\A$ is defined by $\sup_T \E\left[\frac{\A(T)}{\text{OPT}(T)}\right]$, where $T$ is an instance of the online problem.

For the same instance, different deterministic algorithms may have very different ratios. By carefully assigning probabilities to these algorithms, a randomized algorithm $\A$ attempts to capture instances that would otherwise yield large expected ratios. 
Changing the probability assigned to one deterministic algorithm may improve the expected ratios on some instances while worsening those on others. 
The challenge is therefore to balance the contributions of all deterministic algorithms so that no instance escapes with an excessively large expected ratio. 
In this process, even a deterministic algorithm with a large competitive ratio may be essential to achieving the optimal balance.

How to distribute probability among deterministic algorithms to capture these instances is closely tied to the structure of the underlying online problem. In particular, it depends on how the adversary can exploit the uncertainty in the input to make different groups of deterministic algorithms perform poorly.

See Figure~\ref{fig:first} as a conceptual example. The x-axis represents the probability distribution of the deterministic algorithms used by the randomized algorithm $\A$, arranged in non-decreasing order of their competitive ratios. The length on the x-axis occupied by each deterministic algorithm represents its probability in $\A$. The y-axis represents the ratio of each deterministic algorithm on a given instance $T$. The gray area represents the possible ratios, and its upper boundary represents the competitive ratio of the corresponding deterministic algorithm.

Let $T^*$ be a worst-case instance that maximizes the expected ratio of $\A$. 
In Figure~\ref{fig:first}, the thin solid curve depicts the ratios of all deterministic algorithms used by $\A$ on $T^*$, while the thin dashed curve depicts their ratios on another instance $T$. 
The thick solid and dashed horizontal lines represent the corresponding expected ratios of $\A$, respectively. 
It can be seen that $T^*$ has a higher expected ratio than $T$, although some deterministic algorithms with high competitive ratios perform much better on $T$ than on $T^*$.

\runtitle{Tail constraint changes the optimal structure.}
Now consider introducing a tail constraint $(\gamma, \delta=0)$. 
Since this constraint requires that the probability of exceeding $\gamma$ is zero for every instance, deterministic algorithms with competitive ratios larger than $\gamma$ cannot be assigned positive probability in the randomized algorithm. 
A natural first attempt to satisfy this constraint is to simply remove these deterministic algorithms from $\A$ and redistribute their probability mass among the remaining ones, say evenly, as illustrated in Figure~\ref{fig:second}.

As shown in Figure~\ref{fig:second}, this simple modification decreases the expected ratio on $T^*$ but increases that on $T$, making $T$ the worse instance of the two.
Thus, removing deterministic algorithms with poor worst-case guarantees may change which instance is worst for the randomized algorithm. 
More importantly, it may destroy the balance between the deterministic algorithms that made the original distribution effective. 
This illustrates why finding an optimal randomized algorithm under a tail constraint may require more than simply removing the deterministic algorithms that violate the constraint.
The optimal randomized algorithm may have a substantially different structure.

\hide{Now consider introducing a tail constraint $(\gamma, \delta=0)$. Since the tail constraint restricts that $\Pr\Big(\frac{\A(T)}{\text{OPT}(T)} > \gamma\Big) \leq \delta = 0$, any deterministic algorithm with competitive ratio larger than $\gamma$ should be excluded.
Imagine the same randomized algorithm $\A$ after removing all such deterministic algorithms and evenly spreading the spare distribution over the legal ones. The modified randomized algorithm is depicted in Figure~\ref{fig:second}, where the expected ratio of $T^*$ is now smaller than the expected ratio of some other instance $T$.}

\begin{figure}[t]
    \centering
    \begin{subfigure}[b]{0.48\textwidth}
        \centering
        \includegraphics[width=\linewidth]{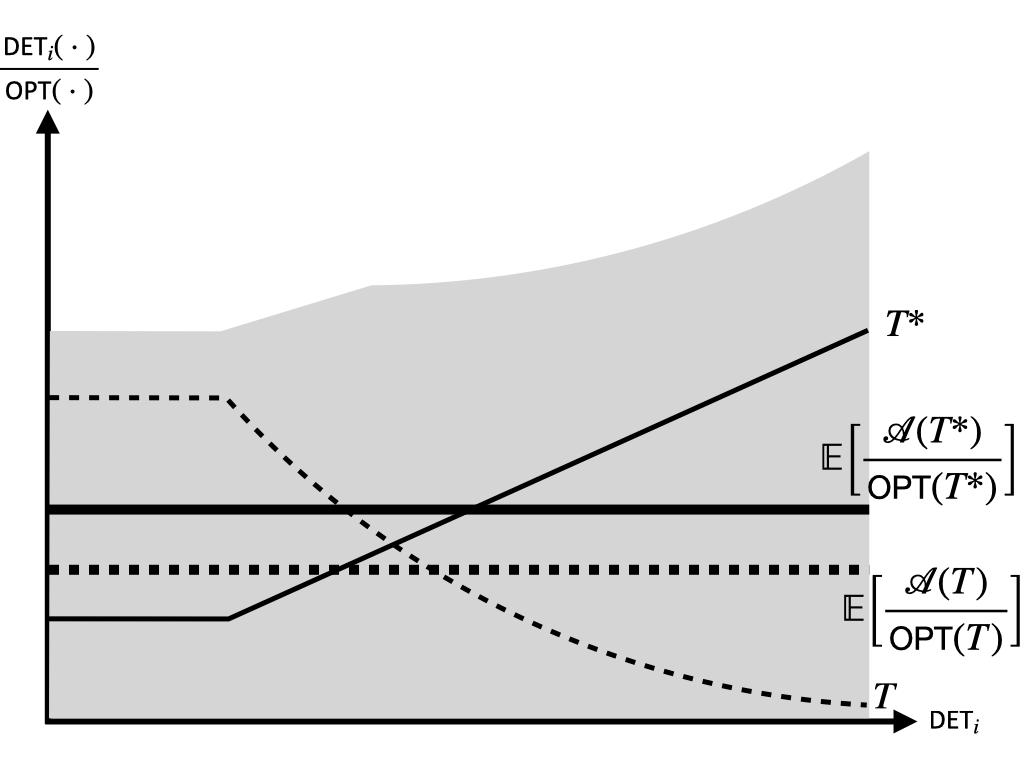}
        \caption{}
        \label{fig:first}
    \end{subfigure}
    \hfill
    \begin{subfigure}[b]{0.48\textwidth}
        \centering
        \includegraphics[width=\linewidth]{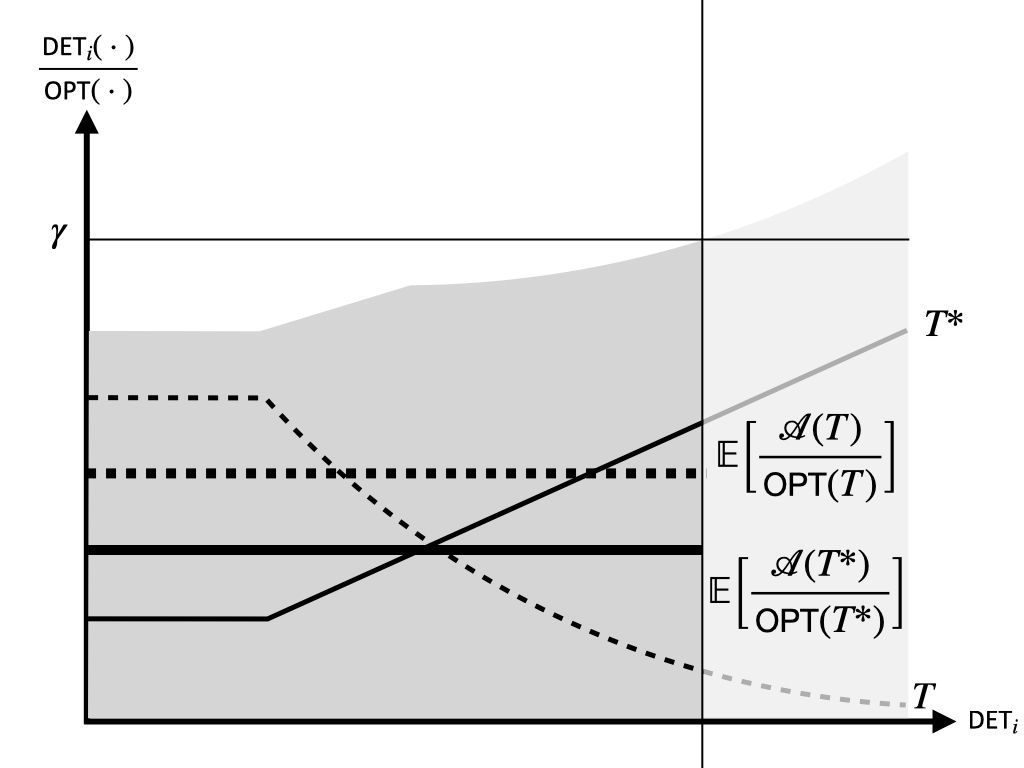}
        \caption{}
        \label{fig:second}
    \end{subfigure}
    \caption{}
    \label{fig:both}
\end{figure}

\medskip

\runtitle{Previous work.}
The model of \emph{online algorithms with tail constraints} was recently introduced by Dinitz et al.~\cite{DBLP:conf/soda/DinitzILMV24}. They initiated this line of research for the ski rental problem and showed that imposing tail constraints can substantially change the structure of optimal randomized algorithms. Cui and Dinitz later extended this model to two-slope ski rental~\cite{DBLP:conf/waoa/CuiD25}.

The concept of tail constraints is closely related to risk measures studied in economics and finance. In particular, \emph{Value-at-Risk (VaR)} concerns the minimum threshold such that the probability of exceeding this threshold is bounded from above by a prescribed probability level. \emph{Conditional Value-at-Risk (CVaR)}, on the other hand, considers the average loss among the worst outcomes beyond a given quantile. Christianson et al.~\cite{DBLP:conf/colt/Christianson0LW24} considered a CVaR-based notion of competitive ratio, which captures the average competitive ratio in the upper tail of its distribution, as determined by a prescribed probability level.

Online bidding was formalized by Chrobak et al.~\cite{chrobak2008incremental} as a framework for studying doubling-based algorithms, with applications to incremental medians.
The classical online bidding problem admits an optimal deterministic competitive ratio of $4$ and an optimal randomized expected competitive ratio of $e$.
Chrobak et al. established the matching randomized lower bound using a linear-programming primal-dual approach, which we extend to the setting with tail constraints.

To the best of our knowledge, prior to our work and the independent concurrent work of Basiak et al.~\cite{BasiakBienkowskiShinTatarczuk26}, tail constraints had only been studied for ski rental and its variants.
These works represent the first extensions of the model beyond ski rental. In particular, Basiak et al. independently established the optimal expected competitive ratio for online bidding under pure tail constraints, using Yao's min-max principle for their lower bound.

\runtitle{Our contributions.}
Our main contribution is a primal-dual approach to understanding randomized online bidding under tail constraints.
By revisiting the classical primal-dual lower-bound framework of Chrobak et al.~\cite{chrobak2008incremental}, we uncover a key structural insight:
the lower-bound argument becomes strictly stronger when the ratio between consecutive bids is bounded from above by a constant strictly smaller than $e$ (Observation~\ref{obs:stronger_lower_bound}).
This observation reveals a promising direction for extending the primal-dual framework to tail-constrained settings.
The central challenge, however, is to establish a connection between the growth of consecutive bids and the tail constraints, which restrict the probability of large ratios.
We bridge this gap by showing that tail constraints with zero violation probability impose an asymptotic upper bound on consecutive-bid ratios (Lemma~\ref{lm:growth_rate_ub} and~\ref{lm:growth_rate}).
Incorporating this structural property into the primal-dual framework allows us to derive stronger lower bounds under tail constraints (Lemma~\ref{lm:delta_0_lb}).

Beyond strengthening the lower-bound analysis, our primal-dual approach also provides insight into the design of optimal bidding strategies (Observation~\ref{obs:with_tail_constraint}).
In particular, the structure of the dual solution suggests a natural geometric growth factor for consecutive bids, motivating a randomized periodic bidding algorithm.
By combining this algorithm with our refined primal-dual lower bound, we obtain a tight characterization of the optimal expected competitive ratio under zero-violation tail constraints (Theorem~\ref{thm:special_case}).
We further extend the algorithmic approach to general tail constraints, establishing upper bounds on the expected competitive ratio (Theorem~\ref{thm:general_case}).
We also show that, for any fixed geometric base, the uniform offset distribution is optimal among randomized periodic bidding algorithms in terms of both the expected competitive ratio and tail-constraint feasibility (Lemma~\ref{lm:choice_of_uniform_offset}).
Thus, when considering periodic algorithms, we can restrict our attention to uniform offset distributions.
Our main results are summarized in the following two theorems.

\begin{theorem}[Optimal ratio for tail constraints with \(\delta_i = 0\)] \label{thm:special_case}
    For a tail constraint $(\gamma_i,0)$ with $\gamma_i \ge 4$, the expected competitive ratio of the online bidding problem is exactly
    \(
    \frac{q_i^*}{\ln q_i^*},
    \)
    where
    \[
    q_i^* =
    \begin{cases}
    \frac{\gamma_i + \sqrt{\gamma_i (\gamma_i - 4)}}{2},
    & \text{if } 4 \le \gamma_i \le \frac{e^2}{e-1}, \\
    e,
    & \text{if } \gamma_i > \frac{e^2}{e-1}.
    \end{cases}
    \]
    For tail constraints $\{(\gamma_i, 0)\}_{i=1}^m$,
    the expected competitive ratio is exactly
    \(
    \frac{q^*}{\ln q^*},
    \)
    where
    \(
        q^*=\min_{i\in[m]} q_i^*.
    \)
\end{theorem}
\begin{theorem}[Upper bound for general tail constraints] \label{thm:general_case}
    For \(q > 1\) and \(\gamma > 1\), let
    \(
        P_q(\gamma)
        :=
        \min \left\{1, \max \left\{0,
            1 - \log_q \left( \frac{\gamma(q-1)}{q} \right)
        \right\} \right\}.
    \)
    Consider a tail constraint \((\gamma_i,\delta_i)\) with \(\gamma_i \ge \gmlb\) and $\delta_i \in (0, 1)$.
    Define
    \[
        q_i^*
        =
        \begin{cases}
            \text{the largest \(q > 1\) satisfying }
            P_q(\gamma_i) = \delta_i,
            & \text{if } 0 < \delta_i \le P_e(\gamma_i), \\
            e,
            & \text{if } \delta_i > P_e(\gamma_i).
        \end{cases}
    \]
    Then the expected competitive ratio of the online bidding problem is at most
    \(
        \frac{q_i^*}{\ln q_i^*}
    \)
    if $q_i^*$ exists.
    More generally, for a collection of tail constraints
    \(
        \{(\gamma_i,\delta_i)\}_{i=1}^m,
    \)
    the expected competitive ratio is at most
    \(
        \frac{q^*}{\ln q^*},
    \)
    where
    \(
        q^*=\min_{i\in[m]}q_i^*,
    \)
    provided that $q_i^*$ exists for all $i$ and this choice satisfies all the constraints.
\end{theorem}

\runtitle{Paper organization.}
In Section~\ref{sec:preliminary}, we define the online bidding problem and introduce the tail constraints.
We then revisit a primal-dual framework for the problem without tail constraints in Section~\ref{sec:previous_primal_dual}.
Section~\ref{sec:delta_0} studies the problem with $\delta_i = 0$.
We introduce a special bidding sequence in Section~\ref{sec:ext_seq} and incorporate its properties into the primal-dual framework in Section~\ref{sec:primal_dual}.
Finally, we discuss the problem with general tail constraints in Section~\ref{sec:general_delta}.
The omitted proofs are provided in the appendix.

\section{Preliminaries} \label{sec:preliminary}

\runtitle{The online bidding problem.}
We consider the online bidding problem, in which an algorithm is faced with an unknown threshold \(T \ge 1\). An algorithm \(\dA\) submits a strictly increasing sequence of bids
\[
    b_1 < b_2 < \cdots
\]
without knowing \(T\). Let \(k\) denote the index of the first bid that reaches or exceeds the threshold, i.e.,
\[
    b_{k-1} < T \leq b_k,
\]
where the condition involving \(b_{k-1}\) is omitted when \(k=1\). The algorithm terminates upon submitting \(b_k\), and incurs cost
\[
    \cost(\dA(T)) \coloneqq \sum_{i=1}^{k} b_i.
\]

Let \(\OPT(T)\) denote an optimal offline algorithm for threshold \(T\). Since \(\OPT\) knows \(T\) in advance, it can submit \(T\) as its first bid and terminate immediately. Hence,
\[
    \cost(\OPT(T)) = T.
\]
For a deterministic algorithm \(\dA\), its \emph{ratio} on an instance \(T\) is
\(
    \frac{\cost(\dA(T))}{\cost(\OPT(T))},
\)
and the algorithm's \emph{competitive ratio} is
\[
    \sup_{T \ge 1}
    \frac{\cost(\dA(T))}{\cost(\OPT(T))}
    =
    \sup_{T \ge 1}
    \frac{\cost(\dA(T))}{T}.
\]
The \emph{competitive ratio of the online problem} is the infimum of the competitive ratios of all deterministic algorithms.

For a randomized algorithm \(\A\), where the expectation is taken over the internal randomness of \(\A\), its \emph{expected ratio} on threshold \(T\) is
\(
    \E\!\left[
        \frac{\cost(\A(T))}{\cost(\OPT(T))}
    \right].
\)
Accordingly, the randomized algorithm's \emph{expected competitive ratio} is
\[
    \sup_{T \ge 1}
    \E\!\left[
        \frac{\cost(\A(T))}{\cost(\OPT(T))}
    \right]
    =
    \sup_{T \ge 1}
    \E\!\left[
        \frac{\cost(\A(T))}{T}
    \right].
\]
The \emph{expected competitive ratio of the online problem} is the infimum of the expected competitive ratios of all randomized algorithms.

The competitive ratio of the online bidding problem is \(4\), whereas the expected competitive ratio of the problem is \(e\) \cite{chrobak2008incremental}.

\runtitle{Tail constraints.}
A tail constraint limits the probability that a randomized algorithm incurs a ratio exceeding a prescribed threshold. To formalize this, we first define the \(\gamma\)-exceedance probability as follows.

\begin{definition}[\(\gamma\)-exceedance probability]
Let \(\A\) be a randomized algorithm. For \(\gamma > \gmlb\), the \emph{\(\gamma\)-exceedance probability} of \(\A\) is
\[
    \sup_{T \ge 1}
    \Pr\left[
        \frac{\operatorname{cost}(\A(T))}{T}>\gamma
    \right],
\]
where the probability is taken over the internal randomness of \(\A\).
\end{definition}

A \emph{tail constraint} is specified by a pair of parameters $(\gamma,\delta)$, where $\gamma > \gmlb$ and $\delta \in [0,1]$.
A randomized algorithm $\A$ \emph{satisfies} the tail constraint $(\gamma,\delta)$ if its $\gamma$-exceedance probability is at most $\delta$, i.e.,
\[
\sup_{T \ge 1}
\Pr\left[
\frac{\operatorname{cost}(\A(T))}{T}>\gamma
\right]
\le \delta.
\]
Equivalently, for every \(T \ge 1\),
\(
    \Pr\left[
        \frac{\operatorname{cost}(\A(T))}{T}>\gamma
    \right]
    \le \delta.
\)

More generally, we consider a collection of tail constraints
\[
    \{(\gamma_i,\delta_i)\}_{i=1}^{m}.
\]
A randomized algorithm $\A$ is \emph{feasible} with respect to this collection if it satisfies all tail constraints simultaneously, i.e., for every \(i\in[m]\) and every $T\geq 1$,
\(
    \Pr\left[
        \frac{\operatorname{cost}(\A(T))}{T}>\gamma_i
    \right]
    \le \delta_i.
\)


Imposing a tail constraint restricts the class of feasible randomized algorithms and may therefore increase the optimal expected competitive ratio.

\subsection{Revisiting the primal-dual lower bound without tail constraints} \label{sec:previous_primal_dual}

Our lower-bound analysis under tail constraints builds on the primal-dual framework of Chrobak et al.~\cite{chrobak2008incremental} for randomized online bidding.
We begin by revisiting their lower-bound argument in the unconstrained setting.
A closer examination of their proof reveals that a strictly stronger lower bound can be obtained if the ratio between every pair of consecutive bids is strictly less than~$e$.
This observation provides a starting point for extending their analysis to the setting with tail constraints.

Chrobak et al.~\cite{chrobak2008incremental} restrict attention to the finite universe
\[
[n]:=\{1,2,\ldots,n\},
\]
so that both the threshold and all bids are required to lie in \([n]\).
They establish the lower bound for this restricted problem with arbitrarily large \(n\).
Since the unrestricted problem includes every finite universe \([n]\) as a special case, any asymptotic lower bound obtained in this setting also applies to the original problem over the real-valued domain.

The authors establish the lower bound using a primal-dual approach.
Specifically, they formulate a factor-revealing linear program whose objective captures the expected competitive ratio.
They then construct a feasible dual solution whose objective value provides a lower bound on the optimal primal objective value by weak duality.
The linear program is formulated as follows.

\begin{figure}[ht]
    \centering
    
    \fbox{%
    \begin{minipage}[t]{0.47\textwidth}
    \footnotesize
    \textbf{Primal}
    \[
        \begin{aligned}
        \text{min}\quad
        & \beta \\
        \text{s.t.}\quad
        & \beta
          - \sum_{t=1}^T \sum_{b=1}^n
            \frac{b}{T} x(t,b)
          \ge 0,
        && \forall T \in [n], \\
        & \sum_{t=1}^T
          \sum_{b=T}^{n}
          x(t,b)
          \ge 1,
        && \forall T \in [n], \\
        & x(t,b) \ge 0,
        && \forall t,b \in [n], \\
        & \beta \ge 0.
        \end{aligned}
    \]
    \end{minipage}
    }
    \hfill
    \fbox{%
    \begin{minipage}[t]{0.47\textwidth}
    \footnotesize
    \textbf{Dual (after normalization)}
    \[
        \begin{aligned}
        \text{max}\quad
        & \frac{\sum_{T=1}^n \mu(T)}{\sum_{T'=1}^n \pi(T')} \\
        \text{s.t.}\quad
        & \sum_{T=t}^b \mu(T)
          -
          \sum_{T=t}^n \frac{b}{T}\pi(T)
          \le 0,\\ 
          &&&\forall t,b \in [n], \\
        & \pi(T) \ge 0,
        && \forall T \in [n], \\
        & \mu(T) \ge 0,
        && \forall T \in [n].
        \end{aligned}
    \]
    \end{minipage}
    }
    
    \caption{Primal and dual linear programs for the classic online bidding problem~\cite{chrobak2008incremental}}
    \label{fig:previous_primal_dual}
\end{figure}

\hide{
\begin{figure}[ht]
\centering

\begin{minipage}[t]{0.49\textwidth}
\begin{tcolorbox}[title=Primal, colback=white, colframe=black]
\footnotesize
\[
    \begin{aligned}
    \text{min}\quad
    & \beta \\
    \text{s.t.}\quad
    & \beta
      - \sum_{t=1}^T \sum_{b=1}^n
        \frac{b}{T} x(t,b)
      \ge 0,
    \forall T \in [n], \\
    & \sum_{t=1}^T
      \sum_{b=T}^{n}
      x(t,b)
      \ge 1,
    \forall T \in [n], \\
    & x(t,b) \ge 0,
    \forall t,b \in [n], \\
    & \beta \ge 0.
    \end{aligned}
\]
\end{tcolorbox}
\end{minipage}
\hfill
\begin{minipage}[t]{0.49\textwidth}
\begin{tcolorbox}[title=Dual (after normalization), colback=white, colframe=black]
\footnotesize
\[
    \begin{aligned}
    \text{max}\quad
    & \frac{\sum_{T=1}^n \mu(T)}{\sum_{T'=1}^n \pi(T')} \\
    \text{s.t.}\quad
    & \sum_{T=t}^b \mu(T)
      -
      \sum_{T=t}^n \frac{b}{T}\pi(T)
      \le 0, \\&\forall t,b \in [n], \\
    & \pi(T) \ge 0, \forall T \in [n], \\
    & \mu(T) \ge 0,
    \forall T \in [n].
    \end{aligned}
\]
\end{tcolorbox}
\end{minipage}

\caption{Primal and dual linear programs.}
\end{figure}
}

\runtitle{The factor revealing linear program~\cite{chrobak2008incremental}.}
Consider a bidding sequence \((b_i)_{i \geq 0}\), where \(b_0=0\) for notational convenience.
Let
\[
    X(t,b)
    =
    \begin{cases}
    1, & \text{if \(t-1\) and \(b\) are consecutive bids in \((b_i)_{i\geq0}\)},\\
    0, & \text{otherwise},
    \end{cases}
\]
for all \(t,b\in[n]\),
and let \(x(t,b)= \E[X(t,b)]\), which can be seen as the probability that $t-1$ and $b$ are consecutive bids in the algorithm's bidding sequence.
Choose a sufficiently large integer \(N\) and set
\[
n= \left\lceil N^2\log N \right\rceil.
\]
The primal linear program is given in Figure~\ref{fig:previous_primal_dual}.

The primal linear program minimizes $\beta$, which represents an upper bound on the expected competitive ratio of the randomized algorithm. The first set of constraints ensures that, for every threshold $T$, the expected total cost is at most $\beta T$. The second set of constraints ensures that the algorithm reaches or exceeds every threshold \(T\) with probability \(1\). Every randomized bidding algorithm induces a feasible solution to this LP. Hence, the optimal LP value provides a lower bound on the optimal expected competitive ratio.

\runtitle{A feasible dual solution and a lower bound of the expected competitive ratio.}
After normalization, the corresponding dual program can be written as the form in Figure~\ref{fig:previous_primal_dual}.
The dual variables are constructed as follows.
\[
\begin{aligned}
\pi(T) &=
\begin{cases}
    \dfrac{1}{T}, & \text{if } N \le T \le N^2 \log N,\\
    0, & \text{otherwise},
\end{cases}
\\[-1ex]
\end{aligned}
\qquad
\begin{aligned}
\mu(T) &=
\begin{cases}
    \dfrac{\alpha}{T}, & \text{if } N \le T \le N^2,\\
    0, & \text{otherwise},
\end{cases}
\end{aligned}
\]

\hide{
\[
    \pi(T) =
    \begin{cases}
      \frac{1}{T}, \quad & \text{if } N \leq T \leq N^2 \log N, \\
      0, \quad & \text{otherwise},
    \end{cases}
\]
\[
    \mu(T) =
    \begin{cases}
      \frac{\alpha}{T}, \quad & \text{if } N \leq T \leq N^2, \\
      0, \quad & \text{otherwise},
    \end{cases}
\]
}
where $\alpha$ is a parameter to be chosen later.

Under this construction, the dual objective value approaches \(\alpha\) as \(N\to\infty\):
\[
    \frac{\sum_{T=1}^n \mu(T)}{\sum_{T=1}^n \pi(T)}
    =
    \alpha\ \frac{H_{N^2} - H_{N-1}}{H_{N^2 \log N} - H_{N-1}}
    \approx_{N \to \infty} \alpha.
\]

Dividing the dual constraint by \(b\) yields
\[
    \sum_{T=t}^n \frac{1}{T}\pi(T)
    \geq
    \sum_{T=t}^b \frac{1}{b} \mu(T).
\]
To establish dual feasibility for a suitable choice of \(\alpha\), they consider three cases according to the value of \(t\).
For $t > N^2$, the right-hand side is zero, and hence the constraint is always satisfied.
For \(t<N\), both \(\pi(T)\) and \(\mu(T)\) are zero for all \(T<N\).
Thus, the constraint is identical to that for \(t=N\), and it suffices to verify the constraint at \(t=N\).

For $N \leq t \leq N^2$, they show that the left-hand side satisfies
\(
    \sum_{T=t}^n \frac{1}{T}\pi(T)
    \geq_{N \to \infty}
    \frac{1}{t}.
\)
The right-hand side can be upper-bounded as follows:
\(
    \sum_{T=t}^b \frac{1}{b} \mu(T)
    \leq \frac{\alpha}{b} \ln \frac{b}{t - 1}.
\)
To determine the worst case, they maximize the function
\begin{equation} \label{eq:b_and_t}
    \frac{1}{b} \ln \frac{b}{t - 1}.
\end{equation}
Let $z$ be defined by $b = z(t - 1)$.
The function is maximized at $z=e$ and the maximum value is
\(
    \frac{1}{e(t-1)}.
\)
Moreover, $\frac{1}{t-1}$ approaches $\frac{1}{t}$ as $t \to \infty$, which is ensured by taking $N \to \infty$.
To ensure that
\(
    \frac{\alpha}{b} \ln \frac{b}{t - 1}
    \leq
    \frac{1}{t},
\)
they choose $\alpha=e$.
Therefore, the constructed solution is feasible.

\runtitle{What we learn from the primal-dual about the optimal randomized algorithm's behavior.}
The maximizer of~\eqref{eq:b_and_t} offers a revealing clue about the primal bidding behavior associated with a worst-case instance.
Intuitively, the dual can be viewed as searching for a worst-case instance, with an optimal dual solution representing such an instance.
By the definition of $X(t,b)$, $t-1$ and $b$ are consecutive bids.
Since~\eqref{eq:b_and_t} is maximized when $b=e(t-1)$, this suggests that consecutive bids in an optimal primal solution should differ by a factor of $e$ on a worst-case instance.
Interestingly, this mirrors the behavior of the corresponding upper bound.

The role of $e$ does not end with this connection.
Equation~\eqref{eq:b_and_t} is increasing for $z\in(1,e)$.
Hence, restricting the ratio of every pair of consecutive bids to at most some $q<e$ yields a strictly stronger lower bound.
This leads to the following observation.

\begin{observation}[Strengthening the primal-dual lower bound] \label{obs:stronger_lower_bound}
Suppose that there exists a constant $q<e$ such that
\[
    \frac{b_{i+1}}{b_i} \leq q
    \qquad \text{for all } i,
\]
where $b_i$ and $b_{i+1}$ are consecutive bids.
Then, for the class of bidding algorithms satisfying this condition, the primal-dual argument of Chrobak et al.~\cite{chrobak2008incremental} yields a strictly stronger lower bound on the expected competitive ratio than the classical lower bound for the unconstrained problem.
\end{observation}

\section{Tail constraints with $\delta_i = 0$} \label{sec:delta_0}

To study the online bidding problem with tail constraints, we begin with the special case where $\delta_i = 0$.
Since $\delta_i = 0$ requires the $\gamma_i$-exceedance probability to be zero, we must have $\gamma_i \ge 4$, as the optimal deterministic competitive ratio is $4$.
We first consider a single tail constraint $(\gamma, 0)$
and then extend our analysis to multiple tail constraints.

As discussed in the previous section, Observation~\ref{obs:stronger_lower_bound} suggests that an upper bound on the ratios between consecutive bids could help us extend the lower bound analysis of the randomized online bidding problem to the setting with tail constraints.
This motivates us to investigate whether the tail constraints imply such an upper bound.
However, the value of this upper bound is not immediately clear.
To derive this bound, we construct a special bidding sequence and use its consecutive-bid ratios to establish an asymptotic upper bound for any feasible bidding sequence.

\subsection{Extremal bidding sequence} \label{sec:ext_seq}

By imposing the tail constraint $(\gamma, 0)$, we show that every realization of a randomized algorithm is $\gamma$-feasible, as defined below.

\begin{definition}[\(\gamma\)-feasible bidding sequence] \label{def:feas_seq}
    Let \((b_1,b_2,\ldots)\) be a bidding sequence, and set \(b_0=0\) for notational convenience. We say that \((b_i)_{i\ge 0}\) is \emph{\(\gamma\)-feasible} if
    \[
    \frac{\sum_{i=1}^{j} b_i}{b_{j-1}+1}\le \gamma
    \qquad\text{for every } j\ge 1.
    \]
\end{definition}

\begin{lemma}[Feasibility of deterministic realizations]
    Let $\A$ be a randomized algorithm satisfying the tail constraint $(\gamma,0)$. Then every deterministic algorithm in the support of $\A$ is $\gamma$-feasible.
\end{lemma}
\begin{proof}
    Suppose, for contradiction, that there exists a deterministic algorithm \(\dA\) in the support of \(\A\) whose bidding sequence \((b_i)_{i \ge 0}\) is not \(\gamma\)-feasible. That is, by Definition~\ref{def:feas_seq}, there exists an index \(j\) such that
    \(
    \frac{\sum_{i=1}^j b_i}{b_{j-1}+1}>\gamma.
    \)
    We show that there exists an instance $T$ on which $\dA$ incurs a ratio larger than $\gamma$. Since $\dA$ is in the support of $\A$, it is selected with positive probability. Therefore, $\A$ also has a positive probability of exceeding $\gamma$ on the same instance, contradicting the tail constraint $(\gamma,0)$.

    Consider the threshold
    \(
    T:=b_{j-1}+1.
    \)
    Since the bids are integral and strictly increasing, we have
    \(
    b_{j-1}<T\le b_j.
    \)
    Hence, \(b_j\) is the first bid of \(\dA\) that reaches or exceeds \(T\), and therefore
    \(
    \cost(\dA(T))=\sum_{i=1}^j b_i.
    \)
    Since \(\cost(\OPT(T))=T=b_{j-1}+1\), it follows that
    \[
    \frac{\cost(\dA(T))}{\cost(\OPT(T))}
    =
    \frac{\sum_{i=1}^j b_i}{b_{j-1}+1}
    >
    \gamma.
    \]
    
    Because \(\dA\) lies in the support of \(\A\), it is selected with positive probability. Consequently,
    \[
    \Pr\!\left[
    \frac{\cost(\A(T))}{\cost(\OPT(T))}>\gamma
    \right]>0,
    \]
    contradicting the assumption that \(\A\) satisfies the tail constraint \((\gamma,0)\). Therefore, every deterministic algorithm in the support of \(\A\) is \(\gamma\)-feasible.
\end{proof}

We consider a particular $\gamma$-feasible bidding sequence that is extremal in the sense that the defining inequality is tight at every step.

\begin{definition}[Extremal bidding sequence]
    The \emph{extremal bidding sequence}, denoted by \((b_i^*)_{i\ge 0}\), is defined by setting \(b_0^*=0\) and requiring
    \[
    \frac{\sum_{i=1}^{j} b_i^*}{b_{j-1}^*+1}
    =
    \gamma
    \qquad\text{for every } j\ge 1.
    \]
\end{definition}

\begin{lemma}[Closed form of the extremal bidding sequence] \label{lm:closed_form}
    For \(\gamma \ge 4\), the extremal bidding sequence \((b_i^*)_{i\ge 0}\) admits the following closed form. If \(\gamma>4\), let
    \(
    s=\sqrt{\gamma(\gamma-4)}.
    \)
    Then, for every \(j\ge 1\),
    \[
    b_j^*
    =
    \frac{\gamma}{s}
    \left[
    \left(\frac{\gamma+s}{2}\right)^j
    -
    \left(\frac{\gamma-s}{2}\right)^j
    \right].
    \]
    If \(\gamma=4\), then, for every \(j\ge 1\),
    \(
    b_j^*=j\,2^{j+1}.
    \)
\end{lemma}
\arxivonly{
    \begin{proof}
        By the definition of the extremal bidding sequence, the $\gamma$-feasibility condition holds with equality for every $j\geq 1$.
    This allows us to derive a linear recurrence relation for the bids, which we then solve using its characteristic polynomial.

    For \(j\ge 1\), let
    \(
    S_j:=\sum_{i=1}^j b_i^*.
    \)
    By the definition of the extremal bidding sequence,
    \(S_j=\gamma(b_{j-1}^*+1)\) for every $j\ge 1$.
    For \(j\ge 2\), using \(S_j=S_{j-1}+b_j^*\),
    we obtain
    \(
    \gamma(b_{j-1}^*+1)
    =
    \gamma(b_{j-2}^*+1)+b_j^*.
    \)
    Rearranging gives the recurrence
    \(b_j^* = \gamma (b_{j-1}^* - b_{j-2}^*)\) for $j \ge 2$,
    with initial conditions
    \(b_0^*=0\)
    and
    \(b_1^*=\gamma.\)
    
    We first consider the case \(\gamma>4\). The characteristic polynomial of the recurrence is
    \(
    x^2-\gamma x+\gamma=0,
    \)
    whose two distinct roots are
    \[
    \frac{\gamma\pm\sqrt{\gamma(\gamma-4)}}{2}.
    \]
    Let
    \(
    s:=\sqrt{\gamma(\gamma-4)}.
    \)
    The general solution is therefore
    \[
    b_j^*
    =
    A\left(\frac{\gamma+s}{2}\right)^j
    +
    B\left(\frac{\gamma-s}{2}\right)^j.
    \]
    The condition \(b_0^*=0\) gives \(A+B=0\), and hence \(B=-A\). Using \(b_1^*=\gamma\), we obtain
    \(
    A\left(
    \frac{\gamma+s}{2}
    -
    \frac{\gamma-s}{2}
    \right)
    =
    As
    =
    \gamma,
    \)
    so
    \(A=\frac{\gamma}{s}\)
    and
    \(B=-\frac{\gamma}{s}.\)
    Thus, for every \(j\ge 1\),
    \(
    b_j^*
    =
    \frac{\gamma}{s}
    \left[
    \left(\frac{\gamma+s}{2}\right)^j
    -
    \left(\frac{\gamma-s}{2}\right)^j
    \right].
    \)
    
    It remains to consider the case \(\gamma=4\). The recurrence becomes
    \(
    b_j^*-4b_{j-1}^*+4b_{j-2}^*=0,
    \)
    whose characteristic polynomial
    \(
    x^2-4x+4=(x-2)^2
    \)
    has a repeated root at \(2\). Hence
    \(
    b_j^*=(A+Bj)2^j.
    \)
    The condition \(b_0^*=0\) gives \(A=0\), while \(b_1^*=4\) gives
    \(
    2B=4,
    \)
    and hence \(B=2\). Therefore
    \(
    b_j^*=j2^{j+1}
    \)
    for every \(j\ge 1\), as claimed.
    \end{proof}
}

The following two lemmas establish an asymptotic upper bound on the ratio between consecutive bids of any $\gamma$-feasible bidding sequence.
The first lemma shows that the asymptotic growth rate of any such sequence is bounded above by that of the extremal bidding sequence, while the second lemma establishes that the \emph{growth rate of the extremal sequence} converges to $G_\gamma :=    \frac{\gamma+\sqrt{\gamma(\gamma-4)}}{2}$, which depends only on $\gamma$.
Together, these results imply that
\[
\limsup_{i\to\infty}\frac{b_{i+1}}{b_i}
\le G_\gamma.
\]

\begin{lemma}[Asymptotic growth comparison] \label{lm:growth_rate_ub}
    For any \(\gamma\)-feasible bidding sequence \((b_i)_{i\ge 0}\),
    \[
    \limsup_{i\to\infty}\frac{b_{i+1}}{b_i}
    \le
    \limsup_{i\to\infty}\frac{b_{i+1}^*}{b_i^*}.
    \]
\end{lemma}
\begin{proof}
    Since shifting the indices by a fixed constant $\lceil\gamma\rceil$
    does not affect the limit superior, it suffices
    to prove that
    \[
    \frac{b_{j+1+\lceil\gamma\rceil}}
         {b_{j+\lceil\gamma\rceil}}
    \le
    \frac{b_{j+1}^*}{b_j^*}
    \qquad \text{for all } j\ge 1.
    \]
    Then, by taking the limit superior on both sides,
    we conclude the lemma.

    Since $(b_i)_{i \ge 0}$ is $\gamma$-feasible,
    for every $j\ge 1$,
    \(
    \sum_{i=1}^{j+1}b_i
    \le \gamma(b_j+1).
    \)
    Rearranging and dividing by $b_j>0$ gives
    \(
    \frac{b_{j+1}}{b_j}
    \le
    \gamma-
    \frac{\sum_{i=1}^{j}b_i-\gamma}{b_j}.
    \)
    Similarly, the definition of the extremal bidding sequence gives
    \(
    \frac{b_{j+1}^*}{b_j^*}
    =
    \gamma-
    \frac{\sum_{i=1}^{j}b_i^*-\gamma}{b_j^*}.
    \)
    Define
    \[
    R_j:=
    \frac{\sum_{i=1}^{j}b_i-\gamma}{b_j}
    \quad
    \text{and}
    \quad
    R_j^*:=
    \frac{\sum_{i=1}^{j}b_i^*-\gamma}{b_j^*}.
    \]
    Then
    \begin{equation}
        \frac{b_{j+1}}{b_j}\le\gamma-R_j
        \quad
        \text{and}
        \quad
        \frac{b_{j+1}^*}{b_j^*}=\gamma-R_j^*.
        \label{eq:ratio-bounds}
    \end{equation}

    Thus, it suffices to establish
    \[
        R_{j+\lceil\gamma\rceil}\ge R_j^*
        \qquad \text{for all } j\ge 1.
    \]
    Once this claim is established, \eqref{eq:ratio-bounds} implies that, for every $j\ge 1$,
    \(
    \frac{b_{j+1+\lceil\gamma\rceil}}
    {b_{j+\lceil\gamma\rceil}}
    \le \gamma-R_{j+\lceil\gamma\rceil}
    \le \gamma-R_j^*
    =\frac{b_{j+1}^*}{b_j^*}.
    \)
    We now prove this claim by induction on $j$.
    
    For $j=1$,
    since $b_1^*=\gamma$, we have
    \(
    R_1^*
    =
    \frac{b_1^*-\gamma}{b_1^*}
    =0.
    \)
    On the other hand, since $b_i\ge 1$,
    \[
    R_{1+\lceil\gamma\rceil}
    =
    \frac{\sum_{i=1}^{1+\lceil\gamma\rceil}b_i
    -\gamma}
    {b_{1+\lceil\gamma\rceil}}
    \ge
    \frac{1+\lceil\gamma\rceil-\gamma}
    {b_{1+\lceil\gamma\rceil}}
    \ge 0
    =R_1^*.
    \]
    Hence the claim holds for $j=1$.
    
    For $j\ge 2$, assume that
    \(
    R_{j-1+\lceil\gamma\rceil}\ge R_{j-1}^*.
    \)
    We show that
    \(
    R_{j+\lceil\gamma\rceil}\ge R_j^*.
    \)
    
    First, 
    the condition $b_i\ge 1$ implies
    \(
    \sum_{i=1}^{j+\lceil\gamma\rceil}b_i\ge j+\lceil\gamma\rceil\ge\gamma,
    \)
    and hence $R_{j+\lceil\gamma\rceil} = \frac{\sum_{i=1}^{j+\lceil\gamma\rceil}b_i-\gamma}{b_{j+\lceil\gamma\rceil}} \ge 0$.

    We next derive recursive relations for $R_{j+\lceil\gamma\rceil}$ and $R^*_j$.
    From the definition of $R_{j+\lceil\gamma\rceil}$,
    \[
    \begin{aligned}
    R_{j+\lceil\gamma\rceil}
    &=
    \frac{\sum_{i=1}^{j+\lceil\gamma\rceil}b_i-\gamma}{b_{j+\lceil\gamma\rceil}}
    =
    1+\frac{\sum_{i=1}^{j-1+\lceil\gamma\rceil}b_i-\gamma}{b_{j+\lceil\gamma\rceil}}
    =
    1+\frac{R_{j-1+\lceil\gamma\rceil}}{b_{j+\lceil\gamma\rceil}/b_{j-1+\lceil\gamma\rceil}}
    \ge
    1+\frac{R_{j-1+\lceil\gamma\rceil}}{\gamma-R_{j-1+\lceil\gamma\rceil}},
    \end{aligned}
    \]
    where the last inequality follows from
    $R_{j-1+\lceil\gamma\rceil}\ge 0$ and
    \eqref{eq:ratio-bounds}, which is
    \(
    0<\frac{b_{k+1}}{b_k}\le\gamma-R_k.
    \)
    Similarly, 
    from the definition of $R^*_j$,
    \[
    \begin{aligned}
    R_j^*
    &=
    \frac{\sum_{i=1}^j b_i^*-\gamma}
    {b_j^*}
    =
    1+\frac{\sum_{i=1}^{j-1}b_i^*-\gamma}
    {b_j^*}
    =
    1+\frac{R_{j-1}^*}{b_j^*/b_{j-1}^*}
    =
    1+\frac{R_{j-1}^*}{\gamma-R_{j-1}^*}.
    \end{aligned}
    \]
    
    Now consider the function
    \(
    f(x)=1+\frac{x}{\gamma-x},
    \)
    for $x<\gamma$.
    The function $f$ is strictly increasing
    on $(-\infty,\gamma)$.
    By the induction hypothesis and the monotonicity of $f$, we have
    \(
    f(R_{j-1+\lceil\gamma\rceil}) \ge f(R_{j-1}^*).
    \)
    Therefore,
    \[
    R_{j+\lceil\gamma\rceil}
    \ge
    1+\frac{R_{j-1+\lceil\gamma\rceil}}
    {\gamma-R_{j-1+\lceil\gamma\rceil}}
    \ge
    1+\frac{R_{j-1}^*}{\gamma-R_{j-1}^*}
    =R_j^*.
    \]
    This completes the induction.
\end{proof}

\begin{lemma}[Growth rate of the extremal bidding sequence] \label{lm:growth_rate}
    For \(\gamma\ge 4\),
    \[
    \lim_{i\to\infty}\frac{b_{i+1}^*}{b_i^*}
    =
    G_\gamma,
    \quad\text{where }
    G_\gamma
    :=
    \frac{\gamma+\sqrt{\gamma(\gamma-4)}}{2}.
    \]
\end{lemma}
\begin{proof}
    We consider the two cases \(\gamma=4\) and \(\gamma>4\) separately.
    
    First suppose that \(\gamma=4\). By Lemma~\ref{lm:closed_form},
    \(
    b_j^*=j2^{j+1}.
    \)
    Hence
    \[
    \frac{b_{j+1}^*}{b_j^*}
    =
    \frac{(j+1)2^{j+2}}{j2^{j+1}}
    =
    2\frac{j+1}{j}.
    \]
    Therefore,
    \(
    \lim_{j\to\infty}\frac{b_{j+1}^*}{b_j^*}
    =
    2.
    \)
    Since
    \(
    G_\gamma
    =
    \frac{\gamma+\sqrt{\gamma(\gamma-4)}}{2},
    \)
    we have \(G_4=2\), and thus the claim follows in this case.
    
    Now suppose that \(\gamma>4\), and let
    \(
    s:=\sqrt{\gamma(\gamma-4)}.
    \)
    By Lemma~\ref{lm:closed_form},
    \(
    b_j^*
    =
    \frac{\gamma}{s}
    \left[
    \left(\frac{\gamma+s}{2}\right)^j
    -
    \left(\frac{\gamma-s}{2}\right)^j
    \right].
    \)
    Factoring out the dominant term gives
    \[
    b_j^*
    =
    \frac{\gamma}{s}
    \left(\frac{\gamma+s}{2}\right)^j
    \left[
    1-
    \left(\frac{\gamma-s}{\gamma+s}\right)^j
    \right].
    \]
    Since \(\gamma>4\), we have \(s>0\) and \(s<\gamma\), and therefore
    \(
    0<
    \frac{\gamma-s}{\gamma+s}
    <1.
    \)
    It follows that
    \[
    \left(\frac{\gamma-s}{\gamma+s}\right)^j
    \to 0,
    \qquad\text{as }j\to\infty.
    \]
    
    Consequently,
    \begin{align*}
    \frac{b_{j+1}^*}{b_j^*}
    &=
    \frac{
    \frac{\gamma}{s}
    \left(\frac{\gamma+s}{2}\right)^{j+1}
    \left[
    1-
    \left(\frac{\gamma-s}{\gamma+s}\right)^{j+1}
    \right]
    }{
    \frac{\gamma}{s}
    \left(\frac{\gamma+s}{2}\right)^j
    \left[
    1-
    \left(\frac{\gamma-s}{\gamma+s}\right)^j
    \right]
    }
    =
    \frac{\gamma+s}{2}
    \cdot
    \frac{
    1-
    \left(\frac{\gamma-s}{\gamma+s}\right)^{j+1}
    }{
    1-
    \left(\frac{\gamma-s}{\gamma+s}\right)^j
    }.
    \end{align*}
    Taking \(j\to\infty\), the second factor converges to \(1\), and hence
    \(
    \lim_{j\to\infty}\frac{b_{j+1}^*}{b_j^*}
    =
    \frac{\gamma+s}{2}
    =
    \frac{\gamma+\sqrt{\gamma(\gamma-4)}}{2}
    =
    G_\gamma.
    \)
    This completes the proof.
\end{proof}

\subsection{Adapting the primal-dual approach to tail constraints} \label{sec:primal_dual}

In this section, we incorporate properties of the extremal bidding sequence into the primal-dual approach proposed by Chrobak et al.~\cite{chrobak2008incremental}.

In the previous section, we show that, for any $\gamma$-feasible bidding sequence, the ratio of two consecutive bids is at most $G_\gamma = \frac{\gamma+\sqrt{\gamma(\gamma-4)}}{2}$, provided that the bids are sufficiently large.
We incorporate this property into the linear program mentioned in Section~\ref{sec:previous_primal_dual}.
More precisely, we relax the second constraint by considering only bids whose values are at least some sufficiently large value $N$, and impose an upper bound of $G_\gamma$ on the ratio of consecutive bids.
We enforce this bound by reducing the upper limit of the summation index $b$.
This yields the primal and dual program in Figure~\ref{fig:primal_dual}.

\begin{figure}[t]
    \centering
    \fbox{%
    \begin{minipage}[t]{0.97\textwidth}
        \footnotesize
        \textbf{Primal}
        \[
            \begin{aligned}
            \text{min}\quad
            & \beta \\
            \text{s.t.}\quad
            & \beta
            - \sum_{t=1}^T \sum_{b=1}^n
              \frac{b}{T} x(t,b)
            \ge 0,
            && \forall T \in [n], \\
            & \sum_{t=N+1}^T
            \sum_{b=T}^{\lfloor G_\gamma (t-1) \rfloor}
            x(t,b)
            \ge 1,
            && \forall T \in [n] \setminus [N], \\
            & x(t,b) \ge 0,
            && \forall t,b \in [n], \\
            & \beta \ge 0.
            \end{aligned}
        \]
    \end{minipage}
    }
    \fbox{%
    \begin{minipage}[t]{0.97\textwidth}
        \footnotesize
        \textbf{Dual (after normalization)}
        \[
            \begin{aligned}
            \text{max}\quad
            & \frac{\sum_{T=1}^n \mu(T)}{\sum_{T'=1}^n \pi(T')} \\
            \text{s.t.}\quad
            & \mathbf{1}[N+1 \le t \; \land \; b \le \lfloor G_\gamma(t-1) \rfloor]
              \sum_{T=t}^b \mu(T)
              -
              \sum_{T=t}^n \frac{b}{T}\pi(T)
              \le 0,
            && \forall t,b \in [n], \\
            & \pi(T) \ge 0,
            && \forall T \in [n], \\
            & \mu(T) \ge 0,
            && \forall T \in [n].
            \end{aligned}
        \]
    \end{minipage}
    }
    \caption{Primal and dual programs for the online bidding problem with tail constraint $(\gamma, 0)$}
    \label{fig:primal_dual}
\end{figure}


The construction of the variables is identical to that of Chrobak et al.~\cite{chrobak2008incremental}.
Thus, the objective value remains unchanged.
The analysis of constraint satisfaction is also unchanged, except for the case $N < t \leq N^2$.
In this case, the analysis of the left-hand side remains unchanged,
\(
    \sum_{T=t}^n \frac{1}{T}\pi(T)
    \geq_{N \to \infty}
    \frac{1}{t}.
\)
The relevant term on the right-hand side can also be upper-bounded in the same way,
\(
    \sum_{T=t}^b \frac{1}{b} \mu(T)
    \leq \frac{\alpha}{b} \ln \frac{b}{t - 1}.
\)
However, the maximum value of the function
\(
    \frac{1}{b} \ln \frac{b}{t - 1}
\)
is different.
Let $z$ be defined by $b = z(t - 1)$.
As mentioned in Section~\ref{sec:previous_primal_dual}, the function is maximized at $z = e$ and is increasing on $z \in (1, e)$.
However,
\(
    b \le \lfloor G_\gamma(t-1) \rfloor \leq G_\gamma(t-1),
\)
which implies that $z \leq G_\gamma$.
Therefore, if $G_\gamma > e$, the maximum is attained at $z=e$; otherwise, it is attained at $z=G_\gamma$.
Consequently, the maximum value of the function is
\(
    \frac{1}{e(t-1)}
\)
if $G_\gamma > e$, and
\(
    \frac{\ln G_\gamma}{G_\gamma(t-1)}
\)
otherwise.
Moreover, $\frac{1}{t-1}$ is asymptotically equivalent to $\frac{1}{t}$ as $t \to \infty$, which is ensured by taking $N \to \infty$.

To ensure that
\(
    \frac{\alpha}{b} \ln \frac{b}{t - 1}
    \leq
    \frac{1}{t},
\)
we choose $\alpha=e$ if $G_\gamma > e$, and
\(
    \alpha=\frac{G_\gamma}{\ln G_\gamma}
\)
otherwise.
Therefore, the constraint is satisfied, and the constructed solution is feasible.

\begin{lemma}[Lower bound for tail constraints with \(\delta_i = 0\)] \label{lm:delta_0_lb}
    For a tail constraint \((\gamma_i, 0)\),
    the expected competitive ratio of the online bidding problem is at least
    \(
        \frac{q_i^*}{\ln q_i^*},
    \)
    where
    \[
    q_i^* =
    \begin{cases}
        \frac{\gamma_i + \sqrt{\gamma_i (\gamma_i - 4)}}{2},
        & \text{if } 4 \le \gamma_i \le \frac{e^2}{e-1}, \\
        e,
        & \text{if } \gamma_i > \frac{e^2}{e-1}.
    \end{cases}
    \]
    For a collection of tail constraints \(\{(\gamma_i, 0)\}_{i=1}^m\),
    the expected competitive ratio is at least
    \(
        \frac{q^*}{\ln q^*},
    \)
    where
    \(
        q^*=\min_{i\in[m]} q_i^*.
    \)
\end{lemma}
\begin{proof}
    For \(\gamma_i \geq 4\), by the construction above, there exists a feasible solution with a expected competitive ratio of \(e\) if \(G_{\gamma_i} > e\), where
    \(
        G_{\gamma_i}
        =
        \frac{\gamma_i+\sqrt{\gamma_i(\gamma_i-4)}}{2}.
    \)
    Since \(G_{\gamma_i}>e\), this is equivalent to
    \(
        \gamma_i > \frac{e^2}{e-1}.
    \)
    On the other hand, if \(G_{\gamma_i}\leq e\), there exists a feasible solution with a expected competitive ratio of
    \(
        \frac{G_{\gamma_i}}{\ln G_{\gamma_i}}.
    \)
    
    For multiple tail constraints, we choose the smallest \(q_i^*\) because doing so ensures that all tail constraints are satisfied.
\end{proof}

We show a matching upper bound in the next section.
The primal-dual analysis also leads to the following observation.

\begin{observation} \label{obs:with_tail_constraint}
    For the online bidding problem with a tail constraint $(\gamma, 0)$,
    the primal-dual argument in this section suggests that
    consecutive bids in an optimal primal solution for a worst-case instance
    should differ by a factor of $\min\{G_\gamma, e\}$.
\end{observation}

\section{Tail constraints with general $\delta_i$} \label{sec:general_delta}

In this section, we derive an upper bound on the expected competitive ratio for the online bidding problem under tail constraints.
We first consider a single tail constraint $(\gamma, \delta)$ and then extend our analysis to multiple tail constraints.

Observation~\ref{obs:with_tail_constraint} from the previous section suggests that consecutive bids in a randomized online algorithm should increase by a factor of $q \le e$, even when a tail constraint is imposed.
This motivates the following periodic algorithm.

\runtitle{Periodic algorithm~$\A_q$.}
For a parameter \(q>1\), we consider the following randomized algorithm \(\A_q\). Initially, \(\A_q\) samples a random offset \(x \sim \operatorname{Uniform}[0,1)\), and then submits the sequence of bids
\[
    q^x, q^{1+x}, q^{2+x}, q^{3+x}, \ldots
\]
until one of the bids reaches or exceeds the threshold.
We will explain how to choose \(q\) later.

This family of algorithms generalizes the optimal randomized algorithm for online bidding: setting \(q=e\) recovers the optimal algorithm, i.e., \(\A_e\).
We now derive an upper bound on the competitive ratio of any realization of \(\A_q\), as well as an upper bound on the expected competitive ratio of \(\A_q\).

\begin{lemma}[Competitive ratio of deterministic realizations] \label{lm:comp_ratio}
    Every deterministic algorithm induced by $\A_q$ has competitive
    ratio at most
    \(
        \frac{q^2}{q-1}.
    \)
    Moreover, $\frac{q^2}{q-1}$ is decreasing on $q \in (1, 2)$ and increasing on $q \in (2,e)$.
\end{lemma}
\arxivonly{
    \begin{proof}
        Fix an arbitrary realization \(x\in[0,1)\) of the random offset. The resulting deterministic algorithm submits the bids
    \[
    b_i=q^{i-1+x},\qquad i\ge 1.
    \]
    Fix a threshold \(T \ge 1\), and let \(k\) be the index of the first bid that reaches or exceeds \(T\).
    
    Suppose that \(k=1\). Then the ratio of the deterministic algorithm is
    \(
    \frac{q^x}{T}
    \le
    q^x
    <
    q
    <
    \frac{q^2}{q-1},
    \)
    where we use \(T \ge 1\), \(x<1\), and \(q>1\).
    
    Now suppose that \(k\ge 2\). Since \(k\) is the first index for which \(b_k\ge T\), we have
    \(
    T>b_{k-1}=q^{k-2+x}.
    \)
    The cost incurred by the algorithm is
    \[
    \sum_{i=1}^{k} q^{i-1+x}
    =
    q^x\sum_{i=0}^{k-1}q^i
    =
    q^x\frac{q^k-1}{q-1}.
    \]
    Hence, the ratio of the algorithm is
    \[
    \frac{q^x(q^k-1)}{(q-1)T} 
    <
    \frac{q^x(q^k-1)}
         {(q-1)q^{k-2+x}} 
    =
    \frac{q^2-q^{2-k}}{q-1}
    <
    \frac{q^2}{q-1}.
    \]
    
    The fact that \(\frac{q^2}{q-1}\) is decreasing on $q \in (1, 2)$ and increasing on $q \in (2,e)$ can be seen by differentiating the function with respect to \(q\).
    \end{proof}
}

\begin{lemma}[Expected competitive ratio of \(\A_q\)] \label{lm:exp_comp_ratio}
    The expected competitive ratio of $\A_q$ is at most
    \(
        \frac{q}{\ln q}.
    \)
    Moreover, $\frac{q}{\ln q}$ is decreasing on $q \in (1,e)$.
\end{lemma}
\begin{proof}
    For a fixed threshold \(T\ge 1\), let \(k\) be the smallest positive
    integer such that the \(k\)-th bid reaches or exceeds \(T\).
    The expected ratio of $\A_q$ on $T$ is
    \[
    \E\!\left[\frac{\sum_{i=1}^{k}q^{i-1+x}}{T}\right]
    =
    \E\!\left[
    \frac{q^k-1}{q-1}\frac{q^x}{T}
    \right] 
    \le
    \E\!\left[
    \frac{q^{k+x}}{q-1}\frac{1}{T}
    \right].
    \]
    
    Let
    \(
    t=\log_q T,
    \)
    so that \(T=q^t\).
    Since \(k\) is the integer such that
    \(
    q^{k-1+x} \ge T > q^{k-2+x},
    \)
    we have
    \(
    k-1=\lceil t-x\rceil.
    \)
    Define
    \(
    \varepsilon=k-1+x-t,
    \)
    which measures the excess of \(q^{k-1+x}\) over \(T\) on a base-\(q\) logarithmic scale.
    Then \(0\le \varepsilon<1\), and
    \(
    T=q^t=q^{k-1+x-\varepsilon}.
    \)
    Therefore, the expected ratio is at most
    \[
    \E\!\left[
    \frac{q^{k+x}}{q-1}
    \frac{1}{q^{k-1+x-\varepsilon}}
    \right] 
    =
    \frac{q}{q-1}\,
    \E [q^\varepsilon].
    \]
    
    It remains to determine the distribution of \(\varepsilon\). Since
    \(
    k-1=\lceil t-x\rceil,
    \)
    we obtain
    \(
    \varepsilon
    =
    \lceil t - x \rceil - (t - x).
    \)
    Equivalently,
    \(
    \varepsilon=\{x-t\},
    \)
    where \(\{\cdot\}\) denotes the fractional part. Since \(x\) is uniform
    on \([0,1)\), translation modulo \(1\) preserves the uniform
    distribution, and therefore
    \(
    \varepsilon\sim\operatorname{Uniform}[0,1).
    \)
    Hence, the expected ratio is at most
    \[
    \frac{q}{q-1}\,
    \E [q^\varepsilon] 
    =
    \frac{q}{q-1}
    \int_0^1 q^z\,dz 
    =
    \frac{q}{q-1}
    \frac{q-1}{\ln q} 
    =
    \frac{q}{\ln q}.
    \]

    The fact that \(\frac{q}{\ln q}\) is decreasing on \(q \in (1,e)\) can be seen by differentiating the function with respect to \(q\).
\end{proof}

Given a base $q>1$ and an arbitrary offset distribution density function $f_u$ on $[0,1)$, we denote $\A_{q, f}$ as a randomized periodic algorithm for online bidding where the $i^\text{th}$-bid is $q^{i-1+u}$ and the offset $u \sim f_u$.
We show that the uniform offset distribution is optimal in terms of the expected competitive ratio among all such algorithms.

\begin{lemma}[Optimality of the uniform offset] \label{lm:choice_of_uniform_offset}    
    Given any tail constraint $(\gamma, \delta)$, if there exists a feasible periodic algorithm $\A_{q,f_u}$ that satisfies the tail constraints, the periodic algorithm with uniform offset $\A_q$ is also feasible.
    Moreover, 
    \[
    \sup_{T \geq 1} \frac{\cost(\A_q(T))}{T} \leq 
    \sup_{T \ge 1} \frac{\cost(\A_{q, f_u}(T))}{T}.
    \]
\end{lemma}

\begin{proof}
Fix an offset $u\in[0,1)$ and consider the corresponding deterministic realization $\dA_u$ of the randomized algorithm $\A_{q,f_u}$.
For any threshold $T$, the last bid of $\dA_u$ is $q^{k+u}$ where $q^{k-1+u} < T \leq q^{k+u}$.

\runtitle{Expected ratio of $\dA_u$ on threshold $T$.}
Since the last bid of $\dA_u$ on threshold $T$ is $q^{k+u}$, the total cost of $\dA_u(T) = q^u\cdot\frac{q^{k+1}-1}{q-1}$. 
Let $T = q^{k-1+y}$, where $y \in (u, u+1]$.
The ratio 
\[
\frac{\cost(\dA_u(T))}{T} = \frac{q^u\cdot\frac{q^{k+1}-1}{q-1}}{q^{k-1+y}}\approx\frac{q}{q-1}\cdot q^{u+1-y}.
\]

\runtitle{Expected competitive ratio of $\A_{q,f}$.}
Now we derive the expected competitive ratio of $\A_{q,f}$, where $f$ is an arbitrary offset density.
First, given a fixed threshold $T$, the expected ratio of $\A_{q,f_u}$ on $T$,
\[
\E_{u \sim f_u}\!\left[\frac{\cost(\dA_u(T))}{T}\right] = \frac{q}{q-1}\cdot \int_0^1 q^{u+1-y} f_u(u) du.
\]

Consider any legal $T$ with $y\in (u, u+1]$. By the average bound, the competitive ratio of $\A_{q,f_u}$,
\[
\sup_T \E_{u \sim f_u}\!\left[\frac{\cost(\dA_u(T))}{T}\right] \geq \frac{q}{q-1}\int_0^1 f_u(u) \Big(\int_u^{u+1} q^{u+1-y} dy\Big) du = \frac{q}{\ln q}.
\]

By Lemma~\ref{lm:exp_comp_ratio}, the expected competitive ratio of $\A_{q,f_u}$ is at least the expected competitive ratio of the algorithm $\A_q$, where the offset density function is uniform.

\runtitle{Feasibility of $\A_q$.}
Now, assume that $\A_{q,f_u}$ is feasible for the tail constraint $(\gamma, \delta)$. We are going to show that $\A_q$ is also feasible for the same tail constraint.
It suffices to show that, in the worst case, the probability that $\A_{q,f_u}$ exceeds $\gamma$ is at least $P_q(\gamma)$, the probability that $\A_q$ exceeds $\gamma$.

Consider an arbitrarily large threshold $T$. The probability that $\A_{q,f_u}$ exceeds $\gamma$, 
\[
\Pr\left[\frac{\cost(\A_{q,f_u}(T))}{T}>\gamma\right] \approx \Pr_{u\sim f_u}\left[\frac{q}{q-1}\cdot q^{u+1-y} > \gamma\right] = \Pr_{u\sim f_u}\left[y < u+1 - \log_q\frac{\gamma(q-1)}{q}\right],
\]
where $q^{k-1+u} < T = q^{k-1+y} \leq q^{k+u}$.

Since $y\in (u, u+1]$, $\Pr_{u\sim f_u}\left[y < u+1 - \log_q\frac{\gamma(q-1)}{q}\right] = 0$ if $\log_q\frac{\gamma(q-1)}{q}\geq 1$, and $\Pr_{u\sim f_u}\left[y < u+1 - \log_q\frac{\gamma(q-1)}{q}\right] = 1$ if $\log_q\frac{\gamma(q-1)}{q} < 0$.
Otherwise, when $\log_q\frac{\gamma(q-1)}{q} \in [0,1)$, the probability that the expected ratio of $\A_{q,f_u}$ exceeds $T$ is $\left(u+1-\log_q \frac{\gamma(q-1)}{q}\right) -u = 1-\log_q \frac{\gamma(q-1)}{q}$. By the three cases, for any arbitrarily large $T$, $\Pr\left[\frac{\cost(\A_{q,f_u}(T))}{T}>\gamma\right]$ is independent on the realized $u$ and is at least $P_q(\gamma)$.
\end{proof}

Since the uniform offset distribution minimizes the expected competitive ratio among all $\A_{q, f}$, we focus on $\A_q$.
We now derive the \(\gamma\)-exceedance probability of \(\A_q\).

\begin{lemma}[\(\gamma\)-exceedance probability of \(\A_q\)] \label{lm:exc_prob}
    For every \(T \ge 1\),
    \[
        \Pr\left[
            \frac{\cost(\A_q(T))}{T}>\gamma
        \right]
        \le P_q(\gamma),
    \]
    where
    \[
        P_q(\gamma)
        =
        \begin{cases}
            1,
            & \text{if } \gamma\le \frac{q}{q-1},\\
            1-\log_q\left(\frac{\gamma(q-1)}{q}\right),
            & \text{if } \frac{q}{q-1}<\gamma<\frac{q^2}{q-1},\\
            0,
            & \text{if } \gamma\ge\frac{q^2}{q-1}.
        \end{cases}
    \]
    Moreover,
    \[
        \lim_{T\to\infty}
        \Pr\left[
            \frac{\cost(\A_q(T))}{T}>\gamma
        \right]
        =
        P_q(\gamma).
    \]
    Consequently, the \(\gamma\)-exceedance probability of \(\A_q\) is \(P_q(\gamma)\).
\end{lemma}
\begin{proof}
    For completeness, we include part of the proof of Lemma~\ref{lm:exp_comp_ratio}.
    Fix a threshold \(T \ge 1\), and write
    \(
    t:=\log_q T.
    \)
    For a realization \(x\in[0,1)\), let \(k\) be the index of the first bid that reaches or exceeds \(T\).
    Since the bids are
    \(
    q^x,q^{1+x},q^{2+x},\ldots,
    \)
    the terminal bid is \(q^{k-1+x}\), where
    \(
    k-1=\lceil t-x\rceil.
    \)
    Indeed, \(t-x>-1\), so \(k-1\) is always a nonnegative integer.
    
    Define the overshoot
    \(
    \varepsilon
    :=
    k-1+x-t
    =
    \lceil t-x\rceil-(t-x).
    \)
    Equivalently,
    \(
    \varepsilon=\{x-t\},
    \)
    where \(\{\cdot\}\) denotes the fractional part. Since \(x\) is uniform on \([0,1)\), translation modulo \(1\) preserves the uniform distribution, and hence
    \(
    \varepsilon\sim\operatorname{Uniform}[0,1).
    \)
    
    For this realization, the competitive ratio is
    \[
    \begin{aligned}
    \frac{\cost(\A_q(T))}{\cost(\OPT(T))}
    &=
    \frac{\sum_{i=0}^{k-1}q^{i+x}}{T}
    =
    \frac{q^x(q^{k}-1)}{(q-1)q^t}\\
    &=
    \frac{q^\varepsilon(q-q^{-(k-1)})}{q-1}
    =
    \frac{q}{q-1}q^\varepsilon
    -\frac{q^x}{(q-1)T}.
    \end{aligned}
    \]

    Since \(1\le q^x<q\), we have
    \[
    \frac{q}{q-1}q^\varepsilon-\frac{q}{(q-1)T}
    \le
    \frac{\cost(\A_q(T))}{\cost(\OPT(T))}
    \le
    \frac{q}{q-1}q^\varepsilon-\frac{1}{(q-1)T}
    \le
    \frac{q}{q-1}q^\varepsilon.
    \]
    In addition,
    \(
    \Pr\!\left[
    \frac{q}{q-1}q^\varepsilon - \frac{q}{(q-1)T}
    >\gamma
    \right]
    =
    \Pr\!\left[
    \frac{q}{q-1}q^\varepsilon
    >\gamma+\frac{q}{(q-1)T}
    \right].
    \)
    Consequently,
    \[
    \Pr\!\left[
    \frac{q}{q-1}q^\varepsilon
    >\gamma+\frac{q}{(q-1)T}
    \right]
    \le
    \Pr\!\left[
    \frac{\cost(\A_q(T))}{\cost(\OPT(T))}>\gamma
    \right]
    \le
    \Pr\!\left[
    \frac{q}{q-1}q^\varepsilon>\gamma
    \right].
    \]

    Since \(\varepsilon\sim\operatorname{Uniform}[0,1)\), we have
    \(
    \frac{q}{q-1}
    \le
    \frac{q}{q-1}q^\varepsilon
    <
    \frac{q^2}{q-1}.
    \)
    Thus, for any \(\gamma'\), the exceedance probability is \(1\) if \(\gamma' \le q/(q-1)\), and \(0\) if \(\gamma' \ge q^2/(q-1)\).
    For the remaining case,
    \(
    \frac{q}{q-1}< \gamma' < \frac{q^2}{q-1},
    \)
    we have
    \[
    \begin{aligned}
    \Pr\!\left[
    \frac{q}{q-1}q^\varepsilon > \gamma'
    \right]
    &=
    \Pr\!\left[
    \varepsilon >
    \log_q\!\left(\frac{\gamma' (q-1)}{q}\right)
    \right]\\
    &=
    1-\log_q\!\left(\frac{\gamma' (q-1)}{q}\right).
    \end{aligned}
    \]
    Therefore, by the definition of \(P_q\),
    \[
    \Pr\!\left[
    \frac{q}{q-1}q^\varepsilon> \gamma'
    \right]
    =P_q(\gamma')
    \]
    for every \(\gamma' \in \mathbb R\). Applying this identity with \(\gamma'=\gamma+\frac{q}{(q-1)T}\) and \(\gamma'=\gamma\) to the probability bounds above yields
    \[
    P_q\!\left(\gamma+\frac{q}{(q-1)T}\right)
    \le
    \Pr\!\left[
    \frac{\cost(\A_q(T))}{\cost(\OPT(T))}>\gamma
    \right]
    \le P_q(\gamma).
    \]

    As \(P_q\) is continuous, letting \(T\to\infty\) and applying the squeeze theorem gives
    \[
    \lim_{T\to\infty}
    \Pr\!\left[
    \frac{\cost(\A_q(T))}{\cost(\OPT(T))}>\gamma
    \right]
    =P_q(\gamma).
    \]

    Together with the upper bound above, this implies
    \(
    \sup_{T \ge 1}
    \Pr\!\left[
    \frac{\cost(\A_q(T))}{\cost(\OPT(T))}>\gamma
    \right]
    =P_q(\gamma).
    \)
    Hence the \(\gamma\)-exceedance probability of \(\A_q\) is \(P_q(\gamma)\).
\end{proof}

\runtitle{Proof of Theorem~\ref{thm:general_case}.}
    Consider a tail constraint $(\gamma, \delta)$, where $\gamma \ge \gmlb$ and $\delta \in (0, 1)$.
    By Lemma~\ref{lm:exc_prob}, the $\gamma$-exceedance probability is given by
    \(
        P_q(\gamma).
    \)
    
    For $0 < \delta \le P_e(\gamma)$, since $\gamma \ge \gmlb$, the relevant solutions to $P_q(\gamma) = \delta$ lie in the range $(1,e]$.
    By Lemma~\ref{lm:exp_comp_ratio}, the expected competitive ratio decreases as $q$ increases over $(1,e)$.
    Therefore, we select $q^*$ as the largest solution to $P_q(\gamma) = \delta$ in $(1,e]$, if such a solution exists.
    
    For $P_e(\gamma) < \delta < 1$, under the same assumption $\gamma \ge \gmlb$, the equation $P_q(\gamma) = \delta$ may admit two solutions, $q_1$ and $q_2$, satisfying $P_{q_1}(\gamma) = P_{q_2}(\gamma) > P_e(\gamma)$.
    However, since $P_e(\gamma) < \delta$, choosing $q=e$ satisfies the tail constraint.
    By Lemma~\ref{lm:exp_comp_ratio}, $q=e$ minimizes the expected competitive ratio.
    Thus, we set $q^*=e$.
    
    In either case, the expected competitive ratio is $\frac{q^*}{\ln q^*}$.
    
    For multiple tail constraints, we choose $q^* = \min_i q_i^*$, provided that this choice satisfies all the constraints.
    \qed

\begin{corollary}[Upper bound for tail constraints with \(\delta_i=0\)] \label{cor:delta_0_ub}
    For a tail constraint $(\gamma_i,0)$ with $\gamma_i \ge 4$, the expected competitive ratio of the online bidding problem is at most
    \(
        \frac{q_i^*}{\ln q_i^*},
    \)
    where
    \[
        q_i^* =
        \begin{cases}
            \frac{\gamma_i+\sqrt{\gamma_i(\gamma_i-4)}}{2},
            & \text{if } 4 \le \gamma_i \le \frac{e^2}{e-1}, \\
            e,
            & \text{if } \gamma_i > \frac{e^2}{e-1}.
        \end{cases}
    \]

    For tail constraints $\{(\gamma_i,0)\}_{i=1}^m$,
    the expected competitive
    ratio is at most
    \(
        \frac{q^*}{\ln q^*},
    \)
    where
    \(
        q^* = \min_{i\in[m]} q_i^*.
    \)
\end{corollary}
\begin{proof}
    By Lemma~\ref{lm:exc_prob}, the $\gamma$-exceedance probability is given by
    \(
        P_q(\gamma).
    \)
    Since $\delta = 0$, we require $P_q(\gamma) = 0$, which is equivalent to
    \(
        q^2 - \gamma q + \gamma \le 0.
    \)
    For $4 \le \gamma \le \frac{e^2}{e-1}$, the largest feasible value of $q$ is the larger root of the corresponding quadratic equation, given by
    \(
        q^*
        =
        \frac{\gamma + \sqrt{\gamma(\gamma-4)}}{2}.
    \)
    On the other hand, by Lemma~\ref{lm:comp_ratio}, the competitive ratio of $\A_e$ is at most $\frac{e^2}{e-1}$.
    Hence, if $\gamma > \frac{e^2}{e-1}$, $\A_e$ satisfies the tail constraint.
    Furthermore, by Lemma~\ref{lm:exp_comp_ratio}, $\A_e$ minimizes the expected competitive ratio.
    Thus, we choose $q^* = e$ in this case.
    
    For multiple tail constraints, we choose $q^* = \min_i q_i^*$ to ensure that all tail constraints are satisfied.
\end{proof}

Corollary~\ref{cor:delta_0_ub} and Lemma~\ref{lm:delta_0_lb} together prove Theorem~\ref{thm:special_case}.

\runtitle{AI Tools Usage Disclosure.}
ChatGPT was used to assist with English language polishing, including grammar correction, sentence clarity, and readability,
as well as with determining the range of $\gamma_i$ in Theorem~\ref{thm:general_case}.

\bibliography{references}

@article{BasiakBienkowskiShinTatarczuk26,
  author  = {Mateusz Basiak and Marcin Bie{\'n}kowski and
             Yongho Shin and Agnieszka Tatarczuk},
  title   = {Pure Tail Constraints for Online Problems},
  journal = {CoRR},
  volume  = {abs/2609.35337},
  year    = {2026},
  eprint  = {2609.35337},
  archivePrefix = {arXiv},
  url     = {https://arxiv.org/abs/2609.35337}
}

@inproceedings{DBLP:conf/soda/DinitzILMV24,
  author    = {Michael Dinitz and
               Sungjin Im and
               Thomas Lavastida and
               Benjamin Moseley and
               Sergei Vassilvitskii},
  title     = {Controlling Tail Risk in Online Ski-Rental},
  booktitle = {Proceedings of the 2024 ACM-SIAM Symposium on Discrete Algorithms
               (SODA)},
  pages     = {4247--4263},
  publisher = {SIAM},
  year      = {2024},
  doi       = {10.1137/1.9781611977912.147}
}

@inproceedings{DBLP:conf/waoa/CuiD25,
  author    = {Qiming Cui and
               Michael Dinitz},
  title     = {Controlling Tail Risk in Two-Slope Ski Rental},
  booktitle = {Workshop on Approximation and Online Algorithms (WAOA)},
  pages     = {81--96},
  year      = {2025}
}

@inproceedings{DBLP:conf/colt/Christianson0LW24,
  author    = {Nicolas Christianson and
               Bo Sun and
               Steven H. Low and
               Adam Wierman},
  title     = {Risk-Sensitive Online Algorithms (Extended Abstract)},
  booktitle = {Conference on Learning Theory (COLT)},
  pages     = {1140--1141},
  year      = {2024}
}

@article{chrobak2008incremental,
  title={Incremental medians via online bidding},
  author={Chrobak, Marek and Kenyon, Claire and Noga, John and Young, Neal E},
  journal={Algorithmica},
  volume={50},
  number={4},
  pages={455--478},
  year={2008},
  publisher={Springer}
}

\end{document}